\documentclass[12pt, reqno]{amsart}

\usepackage{amssymb, amsmath, amsthm,mathrsfs}
\usepackage[backref]{hyperref}
\usepackage[alphabetic,backrefs,lite]{amsrefs}
\usepackage{amscd}   
\usepackage{fullpage}
\usepackage[all]{xy} 
\usepackage{setspace}
\usepackage{mathtools}

\DeclareFontEncoding{OT2}{}{} 

\usepackage[usenames,dvipsnames]{color}

\newtheorem{lemma}{Lemma}[section]
\newtheorem{theorem}[lemma]{Theorem}
\newtheorem{proposition}[lemma]{Proposition}

\newtheorem{claim*}{Claim}

\newtheorem{example}[lemma]{Example}

\theoremstyle{definition}
\newtheorem{remark}[lemma]{Remark}

\newcommand{\Aff}{{\mathbb A}}
\newcommand{\G}{{\mathbb G}}

\newcommand{\PP}{{\mathbb P}}

\newcommand{\F}{{\mathbb F}}

\newcommand{\Z}{{\mathbb Z}}

\newcommand{\calC}{{\mathcal C}}

\newcommand{\calE}{{\mathcal E}}

\newcommand{\calL}{{\mathcal L}}

\newcommand{\calO}{{\mathcal O}}

\newcommand{\frakd}{{\mathfrak d}}

\newcommand{\frakm}{{\mathfrak m}}

\newcommand{\hideqed}{\renewcommand{\qed}{}}

\DeclareMathOperator{\im}{im}

\DeclareMathOperator{\Pic}{Pic}

\DeclareMathOperator{\ev}{ev}

\DeclareMathOperator{\id}{id}

\DeclareMathOperator{\opt}{opt}

\numberwithin{equation}{section}
\numberwithin{table}{section}

\usepackage{color}

\newcommand{\defi}[1]{\textsf{#1}} 

\title[Codes on Surfaces]{Locally recoverable codes on surfaces}

\author{Cec\'ilia Salgado}
\address{Instituto de Matem\'atica, Universidade Federal do Rio de Janeiro, Ilha do Fund\~ao, 21941-909, Rio de Janeiro, Brasil}
\email{salgado@im.ufrj.br}
\urladdr{http://www.im.ufrj.br/\~{}salgado}

\author{Anthony V\'arilly-Alvarado}
\address{Department of Mathematics MS 136, Rice University, 6100 S.\ Main St., Houston, TX 77005, USA}
\email{av15@rice.edu}
\urladdr{http://math.rice.edu/\~{}av15}

\author{Jos\'e Felipe Voloch}
\address{School of Mathematics and Statistics, University of Canterbury, Private Bag 4800, Christchurch 8140, New Zealand}
\email{felipe.voloch@canterbury.ac.nz}
\urladdr{http://www.math.canterbury.ac.nz/\~{}f.voloch}

\date{}

\begin{document}

	\begin{abstract}
	A linear error correcting code is a subspace of a finite-dimensional space over a finite field with a fixed coordinate
	system.
Such a code is said to be locally recoverable with locality $r$ if, 
for every coordinate, its value at a codeword can be deduced from the 
value of (certain) $r$ other coordinates of the codeword. These codes have found many recent
applications, e.g., to distributed cloud storage. We will discuss the problem of
constructing good locally recoverable codes and present some constructions
using algebraic surfaces that improve previous constructions and sometimes provide codes that are optimal in
a precise sense. 
The main conceptual contribution of this paper is to consider surfaces
fibered over a curve {in such a way that each recovery set is constructed from points in a single fiber.}
This allows us to use the geometry of the fiber to guarantee
the local recoverability and use the global geometry of the surface
to get a hold on the standard parameters of our codes. We look in detail 
at situations where the fibers are rational or elliptic curves and provide
many examples applying our methods.

	\end{abstract}

	\maketitle
	

\section{Introduction}

{Motivated by applications to distributed cloud storage, Gopalan et.\ al.~\cite{GopalanHuangEtAl2012} introduced a particular class of error correcting codes that efficiently correct erasures, known now as \emph{locally recoverable codes}.  The successful application of algebraic geometry to the construction of error-correcting codes~\cite{AGBook} naturally prompted the search for locally recoverable codes using algebro-geometric methods \cites{Malmskog2016,BargTamoEtAl2015,BargEtAl2017,Xing,MT,GHM,MTT}.  In particular, \cite{BargTamoEtAl2015} gave a systematic way to produce optimal locally recoverable codes.

{Algebro-geometric} codes are constructed from algebraic varieties, but the one-dimensional case of curves is the most amply studied, while surfaces and higher dimensional varieties have received less attention.  The purpose of this article is to present new systematic constructions of locally recoverable codes using surfaces
fibered over a curve in such a way that each recovery set is constructed from points in a single fiber. We use the geometry of the fiber to guarantee local recoverabity the global geometry of the surface
to get a hold on the standard parameters of our codes. We start by setting up
a general framework for such constructions and showing {how} some previous 
constructions of locally recoverable codes fall into this framework.
We then specialize our setup to consider codes constructed
 using ruled surfaces and elliptic surfaces.  Some of the examples we produce are optimal (in the sense of achieving equality in inequality \ref{dist-bound}) and are long in the sense that they have, for example, length $n = 4q$, where $q$ is the size of the alphabet.
These codes are longer than the other known explicit codes with same recoverability and dimension; {however, they} have bounded recoverability.
We obtain the following theorem as a corollary of Theorem \ref{prop:codesP1P1} (see Example \ref{example-4q}). 
\begin{theorem}
For an integer $d$ divisible by $4$ and an integer $b \le q$ 
such that $4b\geq d$, {there exist} optimal {locally recoverable} codes {over $\F_q$} with parameters
\[
(n,k,d,r) = \left(4b, 3b - \frac{3}{4}d + 1, d, 3\right).
\]
\end{theorem}

For arbitrarily large recoverability, we construct, for every prime $p$,
codes over $\F_{p^2}$ of recoverability $p$, length $n$ about $2p^2$, distance $d$ for any
$d \le n, (p+1)\mid d$, having dimension just shy of the optimal $p(n-d)/(p+1)$.
The precise statement is as follows (Theorem \ref{ell-thm})
\begin{theorem}
For every odd prime (power) $p$ and integer $d \le 2(p+1)(p-2), (p+1)\mid d$,
there exists a locally recoverable code $\calC$ over $\F_{p^2}$ of recoverability $p$, length $n=2(p+1)(p-2)$,
minimum distance $d$ and dimension
$$k = \frac{p(n-d)}{p+1}-\frac{p-1}{2}.$$
\end{theorem}

We believe that codes in this range are new.}

{Ultimately, the codes we construct are obtained by evaluating functions on a curve lying on a surface, and thus can be viewed as codes on curves. However, our proofs of the various properties these codes enjoy crucially rely on the internal geometry of the ambient surface.  This point of view guided our work throughout, so we have kept the perspective it affords.}

{The work of \cite{MT} also uses curves embedded in higher dimensional varieties to 
construct locally recoverable codes. Their construction has some similarities
and some differences to ours. We compare the two constructions, once we set up
some terminology, in Section \ref{mt}
}
\subsection{Locally recoverable codes}
	
	Let $\F_q$ be the finite field of $q$ elements. A linear error correcting code is a subspace $\calC$ of $\F_q^n$ for
	some $n$, which is called the length of $\calC$. We denote by $k$ the dimension of $\calC$ as a $\F_q$-vector space and
	we denote by $d$ the minimum distance of $\calC$, defined as the minimum number of nonzero coordinates among the nonzero elements of $\calC$.
	
	The code $\calC$ is said to be \defi{locally recoverable (LR)} with locality $r$ if, for each $i=1,\ldots, n$, there is a subset $J_i \subset \{1,\ldots, n\}-\{i\},
\#J_i = r$ (called the \defi{recovery set}), such that, if we know the values $c_j$ for $j \in J_i$ of the coordinates of any $c \in \calC$, then we can recover $c_i$.
Codes with small locality can be used in distributed storage systems as they
can reconstruct data erasures with smaller storage overhead than traditional back-ups. 
It is desirable to have codes with small locality, large dimension (equivalently, high information rate $k/n$) 
and large minimum distance for these applications. However, these parameters are not independent: {they satisfy the basic constraint} \cite{GopalanHuangEtAl2012, Dimakis}
\begin{equation}
\label{dist-bound}
d \le n -k - \lceil k/r \rceil +2, 
\end{equation}
and $\calC$ is called an \defi{optimal LR code} if equality holds. We write $d_{\opt}$ for the right hand side of~\eqref{dist-bound}.

An explicit construction of optimal LR codes with $n \le q$ is given in \cite{TamoBarg2014}. 
There are known upper bounds for the length of LR codes and some
general existence theorems \cite{Guruswami2018}.
One of the purposes of this paper is to explicitly construct longer optimal LR codes.

The LR codes we construct have the property that
the sets $J_i \cup \{i\}$ form
a partition of $\{1,\ldots, n\}$ but not every LR code has this property.
We end this subsection by giving a simple proof of
\eqref{dist-bound} for LR codes with this property.

\begin{theorem}
Consider an [n,k,d]-LR code of locality $r$ {whose recovery sets
$J_i$ have the property that the union of the sets} $J_i \cup \{i\}$ form
a partition of $\{1,\ldots, n\}$. Then \eqref{dist-bound} holds.
\end{theorem}

\begin{proof}
Note that the recovery map for any coordinate on inputs all equal to $0$ is $0$, since the
zero vector is a codeword. Now take $b = \lceil k/r \rceil - 1$ so $br <k$
and choose $b$ disjoint sets of the form $J_i \cup \{i\}$ and set the $r$ coordinates indexed by
each $J_i$ from this choice to $0$. In addition,
choose $k-1-br$ coordinates outside the union of the chosen $J_i \cup \{i\}$
and set them equal to $0$ as well. Thus, a total of $k-1$ conditions
are imposed and there exists a non-zero codeword satisfying them all as our code has dimension $k$. 
But this non-zero codeword also has zero $i$-th coordinates for all of 
the chosen $J_i \cup \{i\}$. This gives us $b$ additional
zero coordinates. Hence the weight of this codeword is at most
$n -(k-1)-b = n -k -\lceil k/r \rceil +2$.
\end{proof}
					
\subsection{Algebro-geometric codes}
\label{subsec:AGcodes}
	
	 Let $X$ be a quasi-projective \footnote{Many codes are naturally described as algebro-geometric (AG) codes in quasi-projective varieties that are not projective. Witness the classical Reed-Muller codes; they are AG codes in affine space $\Aff^n$. Every quasi-projective variety is an open subset of a projective variety, by definition, but the choice of a projective compactification is not unique, e.g., $\Aff^n$ can be embedded in projective space $\PP^n$ or in the product $(\PP^1)^n$, which are different. In this paper, we use specific choices of compatifications when determining parameters for our codes (e.g., when we consider Hirzebruch surfaces).  In other circumstances, it is preferable not to, e.g., a curve can always be embedded in a unique projective curve without increasing the number of singular points. This is not true in higher dimensions.} algebraic variety over a finite field $\F_q$. Concretely, this means that we select an open subset of affine or projective space where a collection of polynomials vanish. 
	 The function field of $X$ is the set of functions that can be expressed
	 as quotients of polynomials in the coordinates of the ambient space modulo the equations defining $X$. 
	 Given a point $P$ on $X$ and an element
	 $\sigma$ of the function field of $X$, if the denominator of $\sigma$ does not vanish at $P$, 
	 the function $\sigma$ can be evaluated at
	 $P$ giving an element $\sigma(P)$ of $\F_q$.
	 
	 Let $P_1,\dots,P_n$ be a subset of the set $X(\F_q)$ of $\F_q$-rational points of $X$ and $V$ a finite-dimensional subspace of
	 the function field of $X$. We assume that the evaluation, as above, of all elements of $V$ at all the points 
	 $P_1,\dots,P_n$ is defined and we can consider the image $\calC$ of evaluation map, which is an error correcting code:
	\begin{align*}
		\ev_{V}\colon  V&\to (\F_q)^n \\
		\sigma &\mapsto \left(\sigma(P_1),\dots,\sigma(P_n)\right).
	\end{align*}
	The length of the code is $n$. The dimension $k$ of the code is 
	\[
		k = \dim_{\F_q}(\im \ev_{V}) = \dim V - \dim_{\F_q}(\ker \ev_{V})
	\]
	which simplifies to $\dim V$ if $\ev_V$ is injective. The minimum distance $d$ is the smallest Hamming distance between elements of $\calC$. 
	This is equal to $n$ minus the largest number of $\F_q$-points of $X$ vanishing on an element of 
	$V \setminus \ker \ev_{V}$.
	
	For $X$ a projective variety and $D$ a divisor on $X$, we denote by $\calL(X,D)$ the Riemann-Roch space of
	functions $\sigma$ on $X$ such that either $\sigma = 0$ or $(\sigma)+D$ is an effective divisor, where $(\sigma)$
	denotes the divisor of $\sigma$. The space $\calL(X,D)$ is always finite-dimensional and we denote its dimension
	by $\ell(X,D)$. We will typically define our vector space $V$ as above as a subspace of some
	 $\calL(X,D)$.
	
\section{Baseline codes from high-dimensional varieties}
\label{baseline}

{Let $\Aff^m$ denote affine $m$-dimensional space over a finite field $\F_q$.} In this section we construct locally recoverable codes, with local recoverability parameter $r$ from a projection morphism 
\begin{align*}
\pi\colon \Aff^{r-1}\times \Aff^1 &\to \Aff^1,\\
(x_1,\dots,x_{r-1};t) &\mapsto t.
\end{align*}
We shall impose the smallest possible amount of structure on our choice of points for evaluation. This will give us a baseline to assess the parameters of other constructions.

Let $M$ and $N$ denote positive integers. {We shall use the space of functions} 
\[
V[M,N] := \left\{ a_0(t) + \sum_{i = 1}^{r-1} a_i(t)x_i : \deg a_0 \leq M \text{ and }\deg a_i \leq N \text{ for } i = 1,\dots,r-1\right\}
\]
{to construct an evaluation code (so $V[M,N]$ plays the r\^ole of the vector space $V$ from \S\ref{subsec:AGcodes}).} 
We pick, for some $b \le q$, some set of $b$ distinct points on 
the target $\Aff^1$ of the morphism $\pi$ and,
in each of $b$ fibers of $\pi$ above these points, 
we pick $r+1$ points and take all these $b(r+1)$ points as the set of points where we evaluate the above functions. Thus, the length of {the resulting code will be} $n = b(r+1)$. {The following lemma falls within the framework of~\cite[Proposition~4.2]{BargEtAl2017}.}

\begin{lemma}
\label{lem:BaselineRecoverability}
Fix $t = t_0 \in \F_q$, and let $P_1,\dots,P_{r+1}$ be $\F_q$-points in the fiber $\pi^{-1}(t_0)$, no $r$ of which lie on a hyperplane. Let $\sigma \in V[M,N]$ be a function. Then the value of $\sigma(P_i)$ can be recovered from knowledge of the coordinates of $P_1,\dots,P_{r+1}$ and the $r$ values $\sigma(P_1),\dots,\widehat{\sigma(P_i)},\dots,\sigma(P_{r+1})$. 
\end{lemma}

\begin{proof}
Write $\sigma = a_0(t) + \sum_{i = 1}^{r-1} a_i(t)x_i$. Let $a_i = a_i(t_0)$ for $i = 1,\dots,r+1$. Then we have the matrix equation
\begin{equation}
\label{eq:BaselineRecoverability}
\begin{pmatrix}
1 & x_1(P_1) & \cdots x_{r-1}(P_1) \\
  & \vdots & \\
\hat 1 & \widehat{x_1(P_i)} & \cdots \widehat{x_{r-1}(P_i)} \\
  & \vdots & \\
1 & x_1(P_{r+1}) & \cdots x_{r-1}(P_{r+1}) \\
\end{pmatrix}
\cdot
\begin{pmatrix}
a_0 \\
a_1 \\
\vdots \\
a_r
\end{pmatrix}
=
\begin{pmatrix}
\sigma(P_1) \\
\vdots \\
\widehat{\sigma(P_i)} \\
\vdots \\
\sigma(P_{r+1})
\end{pmatrix}.
\end{equation}
Since no $r$ of the points $P_1,\dots,P_{r+1}$ lie on a hyperplane, the $r\times r$ matrix in~\eqref{eq:BaselineRecoverability} is invertible, and hence we may compute $a_0,\dots,a_r$ from knowledge of the coordinates of $P_1,\dots,\widehat{P_i},\dots,P_{r+1}$ and the $r$ values $\sigma(P_1),\dots,\widehat{\sigma(P_i)},\dots,\sigma(P_{r+1})$. We conclude that
\[
\sigma(P_i) = a_0 + a_1x_1(P_i) + \cdots a_rx_{r-1}(P_i).\tag*{\qed}
\]
\hideqed
\end{proof}

To construct what we will call {a \defi{baseline code}}, let (as above)
\[
\{t_1,\dots,t_b\} \subseteq \Aff^1(\F_q)
\]
be $b$ distinct points on the target $\Aff^1$ of the morphism $\pi$, and for each $t_i$, choose $r+1$ points $P_{i,1},\dots,P_{i,r+1}$ on the fiber $\pi^{-1}(t_i)$, no $r$ of which lie on a hyperplane. 

\begin{proposition}
\label{prop:paramsBaselineCode}
Suppose that $b-M, b - N \geq 1$. The {baseline} code
\[
\calC = \{ (\sigma(P_{i,j}))_{1\leq i \leq b, 1\leq j \leq r+1} : \sigma \in V[M,N]  \}.
\]
has local recoverability $r$ and its parameters satisfy
\begin{align*}
n &= b(r+1), \\
k &= (M+1) + (r-1)(N+1), \\
d &\leq (r+1)\left(b - (N+1)\right) - (M-N) - \left\lceil\frac{M-N}{r}\right\rceil + 2,\\
d &\geq \min\{(b-M)(r+1),2(b - N)\}.
\end{align*}
\end{proposition}

\begin{proof}
We have already discussed the length of $\calC$. The dimension of the code is simply the $\F_q$-dimension of $V[M,N]$. The upper bound on the distance of the code is an application of~\eqref{dist-bound}. For the lower bound on $d$, we argue as follows: Suppose that $\sigma \in V[M,N]$ is a function with $a_i \equiv 0$ for $i = 1,\dots r-1$, i.e., $\sigma = a_0(t)$ for a polynomial $a_0(t)$ of degree $\leq M$. Then at least $(b - M)$ of the values $a_0(t_1),\dots,a_0(t_b)$ are nonzero. The weight of the codeword associated to $\sigma$ is thus at least $(b-M)(r+1)$. If, on the other hand, $\sigma \in S$ is a function where at least one $a_i \not\equiv 0$ for $i = 1,\dots,r-1$, then at least $(b - N)$ of the values $a_i(t_1),\dots,a_i(t_b)$ are nonzero. In the corresponding fibers of $\pi$, the function $\sigma$ defines a hyperplane. The hypothesis that no $r$ points on a fiber lie on a hyperplane ensures that $\sigma$ takes on a nonzero value on at least two points in each of the $(b - N)$ fibers. Hence, $d \geq \min\{(b-M)(r+1),2(b - N)\}$, as claimed.

Local recoverability of $\calC$ follows from Lemma~\ref{lem:BaselineRecoverability}.
\end{proof}

\begin{remark}
\label{rem:exactd}
The proof of Proposition~\ref{prop:paramsBaselineCode} shows that if $\min\{(b-M)(r+1),2(b - N)\} = (b-M)(r+1)$, then in fact $d = (b-M)(r+1)$.  In addition, if
\begin{equation}
\label{eq:dsharp}
M + N > b\quad\text{and}\quad 2N > b
\end{equation}
then it is always possible to construct a function $\sigma$ whose associated code word has weight exactly $2(b - N)$. So under the conditions~\eqref{eq:dsharp}, the lower bound for $d$ in Proposition~\ref{prop:paramsBaselineCode} is in fact sharp.
\end{remark}

\begin{example}
\label{ex:goodcode}
We specialize to the case where $r = 3$, $M = b-1$ and $N = b-2$. Then the upper and lower bounds for $d$ meet and we have $d = 4$. This gives optimally recoverable codes with parameters
\[
(n,k,d,r) = (4b,3b-2,4,3).
\]
Note that the information rate $k/n$ is approximately $75\%$, and since $b \leq q$, one can construct codes with $n = 4q$ and high information rate that are optimal locally recoverable. In particular, over any $\F_q$ with $q\geq 9$, we can construct a code with parameters $(n,k,d,r) = (32,22,4,3)$.
\end{example}

\begin{example}
If we now take $b\le q,r$ arbitrary and $M=N=b-1$, then the upper and lower bounds of Proposition \ref{prop:paramsBaselineCode} also coincide and the code has distance $d=2$.
\end{example}

The last two examples are the only cases where the upper and lower bounds of Proposition \ref{prop:paramsBaselineCode} coincide and a baseline code with no additional properties is optimal (see Remark \ref{rem:exactd}).  

To see this, let $\delta := M - N$; we consider two cases:
\begin{itemize}
	\item $\min\{(b-M)(r+1),2(b - N)\} = (b-M)(r+1)$: Then 
	\[
	(b-M)(r+1) = (r+1)\left(b - (N+1)\right) - (M-N) - \left\lceil\frac{M-N}{r}\right\rceil + 2,
	\]
	from which one can conclude that
	\begin{equation}
	\label{eq:delta1}
	(r+1)(\delta - 1) - \delta - \left\lceil\frac{\delta}{r}\right\rceil + 2 = 0.
	\end{equation}
	Write $ \left\lceil\frac{\delta}{r}\right\rceil = \frac{\delta}{r} + \epsilon, 0\le \epsilon < 1$, then
	\begin{equation}
	\label{eq:delta2}
	\delta(r-1/r) = r-1 + \epsilon.
	\end{equation}
	This implies in particular that $\delta > 0$. If $\delta \ge 2$ then, since $r \ge 3$, we have
	\[
	\delta(r-1/r) \ge 2r -1 > r-1 + \epsilon
	\]
	and hence $\delta = 1$, since it is an integer.

	{If $\delta = 1$ then the hypothesis $2(b - N)\geq (b-M)(r+1)$ gives
	\[
	b \leq M + \frac{2}{r-1} \leq M + 1 \text{ (whenever $r\geq 3$)},
	\]
	from which we conclude that $b - M = 1$, and hence that $b - N = 2$. It follows that}
	\[
	(r+1) = (b - M)(r+1) = d = \min\{(b - M)(r+1),2(b-N)\} = \min\{(r+1),4\},
	\]
	whence $r + 1 \leq 4$. Since we want codes with $r \geq 3$, we must have $r = 3$ and $d = 4$.
	\medskip

	\item $\min\{(b-M)(r+1),2(b - N)\} = 2(b - N)$: 
	
Then 
	\begin{align}\label{bd1}
	 b= N+1 +\frac{\delta}{r-1}+\frac{1}{r-1}\left\lceil\frac{\delta}{r}\right\rceil.
	\end{align}
	On the other hand, $\min\{(b-M)(r+1),2(b - N)\} = 2(b - N)$ gives
	\begin{align}\label{bd2}
	 b\geq \frac{M(r+1)-2N}{r-1}.
	\end{align}
Substituting the value for $b$ obtained in (\ref{bd1}) into the inequality (\ref{bd2}) we get
\begin{align*}
 \delta \left(1+\frac{1}{r-1}\right)\leq 1+ \frac{1}{r-1} \left\lceil\frac{\delta}{r}\right\rceil.
\end{align*}
The latter implies that  $\delta \leq 1$. 
If $\delta =1$ then $M=N+1$ and thus 
\begin{align*}
 2(b-N) &\leq (b-(N+1))(r+1)\\
 \implies b &\geq N+1 +\frac{2}{r-1}.
\end{align*}
We also have 
\begin{align*}
b&=N+1 +\frac{1}{r-1}\left\lceil\frac{1}{r}\right\rceil +\frac{1}{r-1}  \text{ (by~\eqref{bd1})}\\
&= N+1 +\frac{2}{r-1} \\
&\leq N+2.
\end{align*}
The distance is thus given by $d=2(b-N)\leq 4$ and by our analysis, the inequality $2(b-N)\leq (b-(N+1))(r+1)$ is sharp, so $(b-M)(r+1) \leq 4$, which forces $r\leq 3$. Finally, since we assumed $r\geq3$, we conclude that in fact $r=3$ and $d=4$.

If $\delta =0$ then (\ref{bd1}) gives $b=N+1$, which implies that $d=2$. 
If $\delta \leq -1$, we get $b\leq N$ which is not possible. 
\end{itemize}

\section{Codes from ruled surfaces: affine intimations}
\label{s:RuledSurfacesAffine}

\subsection{Tamo-Barg codes}
\label{subsec:Tamo-Barg}

We present the construction of Tamo and Barg \cite{TamoBarg2014} of optimal LR codes of length at most $q$
from the perspective of the last section, which we believe is new. We retain the notation of the previous section.

Let $g(x) \in \F_q[x]$ be a polynomial of degree $r+1$, viewed as a morphism $g\colon \Aff^1 \to \Aff^1$. Choose distinct $t_1,\dots,t_b \in \F_q$ such that the fiber $g^{-1}(t_i)$ consists of
$r+1$ distinct elements $x_{i,1},\ldots,x_{i,r+1}$ of $\F_q$, for $i=1,\ldots,b$. Note that the $x_{i,j}$ are therefore $n=b(r+1)$ distinct elements of $\F_q$. We define the points 
$P_{i,j} = (x_{i,j},x_{i,j}^2,\ldots,x_{i,j}^{r-1}) \in \Aff^{r-1}(\F_q)$, and we consider the projection map
\begin{align*}
\pi\colon \Aff^{r-1}\times \Aff^1 &\to \Aff^1,\\
(x_1,\dots,x_{r-1};t) &\mapsto t.
\end{align*}
For a fixed $i$, the fiber above $t_i$ is an affine space $\Aff^{r-1}$ containing the points $P_{i,j}$ for $j = 1,\dots,r+1$. Moreover, by their construction, these points lie on an \emph{affine rational normal curve}, i.e., they lie on the image of the map
{
\begin{align*}
h\colon \Aff^{1} &\to \Aff^{r-1},\\
x &\mapsto (x,x^2,\dots,x^{r-1}).
\end{align*}
}
This guarantees that no $r$ of them lie on a hyperplane. As in \S\ref{baseline}, we take the space of functions $V[M,N]$, but specialize to the case where $M = N$, and build a code $\calC$. Lemma~\ref{lem:BaselineRecoverability} guarantees that $\calC$ has local recoverability $r$. Put differently, the fact that the points $P_{i,j}$ lie on rational normal curves implies that the $r\times r$ matrix in~\eqref{eq:BaselineRecoverability} is a Vandermonde matrix, thus invertible. 

The parameters $n$, $k$, and $r$ for the code $\calC$ are as before. However, in this special situation, we get a better lower bound for the minimum distance $d$ as follows. Note that 
\[
\sigma(P_{i,j}) = a_0(g(x_{i,j})) + \sum_{\ell=1}^{r-1} a_{\ell}(g(x_{i,j}))x_{i,j}^{\ell}
\]
is the value at $x=x_{i,j}$
of a polynomial of degree at most $N(r+1)+r-1$ in $x$. This degree is an upper bound on the number of its zeros and thus
$d \ge n - (N(r+1)+r-1)$. On the other hand, as in the previous section, the upper bound ~\eqref{dist-bound} for $d$ {when $M = N$} is 
$$(r+1)\left(b - (N+1)\right) - (M-N) - \left\lceil\frac{M-N}{r}\right\rceil + 2 = n - (N(r+1)+r-1),$$ 
showing that these codes are optimal LR codes. 

As mentioned above, these codes have $n \le q$. To achieve $n$ near $q$ one needs to choose the polynomial $g(x)$ in
such a way that the preimage of many values of $t \in \F_q$ under $g$ consists of $r+1$ elements of $\F_q$. One such choice is
$g(x)=x^{r+1}$ if $(r+1)\mid(q-1)$. For other choices and a full discussion, see \cite{TamoBarg2014}.

\subsection{Ruled surfaces perspective}
\label{ss:ruledperspective}

{An algebraic surface $S$ over a field $k$ is called a \defi{ruled surface} if it is endowed with a morphism $\pi\colon S \to B$ to a base algebraic curve $B$ such that for all but finitely many $b \in B(\bar{k})$, the fiber $\pi^{-1}(b)$ is a smooth rational curve, where $\bar{k}$ is a fixed algebraic closure of $k$. There is a ruled surface operating behind the scenes in our 
recasting of the Tamo--Barg codes~\cite{TamoBarg2014}, which we now describe.

Using the notation of \S\ref{subsec:Tamo-Barg}, we let\footnote{Keen readers will immediately note that $S = \Aff^2_{(x,t)}$.  We prefer to use the product $\Aff^1_x \times \Aff^1_t$ because, as we shall see in \S\ref{s:P1xP1}, the correct projective compactification of $S$ to work with is $\PP^1\times\PP^1$, and not $\PP^2$.} $S = \Aff^1_x \times \Aff^1_t$, which maps to $\Aff^{r-1}\times \Aff^1_t$ via 
\[
h\times \id \colon (x,t) \mapsto (x,x^2,\dots,x^{r-1};t).
\]
The variety $S$ fits into the commutative diagram
\[
\xymatrix{
	S \ar[r]^{h\times \id}\ar[d]_{\pi'} & \Aff^{r-1}\times\Aff^1_t \ar[d]^\pi \\
	\Aff^1_x \ar[r]^g & \Aff^1_t
}
\]
where the map $\pi'\colon S \to \Aff^1_x$ is projection onto the first coordinate. The variety $S$ is our ruled surface, and the code constructed in~\S\ref{subsec:Tamo-Barg} can be described as an evaluation code on $S$, as follows. }
Given $t_1,\dots,t_b$ outside the branch locus of the morphism $g\colon \Aff^1_x\to\Aff^1_t$, i.e., such that the fiber $g^{-1}(t_i)$ consists of $b$ distinct points $x_{i,1},\dots,x_{i,r+1}$ in $\Aff^1_x(\F_q)$, we set
{
\[
P_{i,j} = (x_{i,j},t_i) \in S(\F_q) \quad\text{for }1\leq i \leq b, 1\leq j \leq r+1,
\]
}
so that the recovery set for the point $P_{i,j}$ is 
\[
J_{i,j} := \{P_{i,k} : 1 \leq k \leq r+1, k \ne j\}.
\]
Then, letting
{
\[
V[N] = \left\{ a_0(t) + \sum_{i = 1}^{r-1} a_i(t)x^i : \deg a_i \leq N \text{ for } i = 0,\dots,r-1\right\}
\]
}
the Tamo--Barg codes are of the form
\[
\calC = \{ (\sigma(P_{i,j}))_{1\leq i \leq b, 1\leq j \leq r+1} : \sigma \in V[N]  \}.
\]

\subsection{Recasting and extending Barg--Tamo--Vl\u{a}du\c{t} codes}
\label{cyclic}
Just as \S\S\ref{subsec:Tamo-Barg}--\ref{ss:ruledperspective} gives a reinterpretation of the construction of
\cite{TamoBarg2014}, {in this section we reinterpret}
the construction of \cite{BargTamoEtAl2015} but here we go further and, aided
by our geometric point of view, obtain better codes by a judicious choice
of the space of functions to evaluate. Some of the codes we obtain are optimal.

{In broad terms, we consider a \emph{curve} $C$ in the surface $S = \Aff^{1}_x\times\Aff^1_t$ and embed $S$ (and consequently $C$) in $\Aff^{r-1}\times\Aff^1_t$} as above by $(x,t) \mapsto (x,x^2,\ldots,x^{r-1},t)$. We choose $C$ so that the projection in the $t$ coordinate has
degree $r+1$ and choose the values of $t \in \F_q$ to be those for which their preimage consists of
$r+1$ rational points. {Then, just as before, we can evaluate these} points on a space of polynomials 
similar to the ones considered above to get an LR code with locality $r$.

{In \S\ref{ss:ruledperspective} all the points in $S$ used for the Tamo--Barg evaluation code lie on the curve $g(x) = t$. In this section, we instead consider the curve}
\begin{equation}
\label{crv}
C :\quad x^{r+1} = t^2 + 1,
\end{equation}
{which is} a cyclic cover of $\Aff^1_{\F_q}$  via the map $(x,t)\mapsto t$. In order to have many fibers of cardinality $r+1$ over $\F_q$ we take $q \equiv 1\bmod r+1$. {Fix a positive integer $\frakd$. The space of functions we use to define the code consists of functions of the form
\begin{equation}
\label{fn}
\sigma = a_0(t) + a_1(t)x + \cdots + a_{r-1}(t)x^{r-1},
\end{equation}
where the $a_j(t)$ vary in the vector space defined by the inequalities
\[
\deg a_j \le \frac{n-\frakd}{r+1} - \epsilon_j
\]
and
\[
\epsilon_j = 
\begin{cases}
0 & \text{if } j = 0,\\
1 & \text{if } 0 < j \leq (r+1)/2, \\
2 & \text{otherwise.}
\end{cases}
\]
}
The local recoverability with locality $r$ {of the resulting code} follows, since for fixed $t$,
with $r+1$ distinct values for $x$, the matrix determining the missing value is a Vandermonde matrix. The inequalities {defining the space of functions to be evaluated} ensure that the minimum distance of this code is at least $\frakd$, because $x$ has a pole of order $2$ at infinity and $t$ has a pole of order $r+1$ at infinity.

The space of functions at which we evaluate points of the curve has dimension,
for $r$ odd,
\[
k = \frac{r}{r+1}(n-\frakd) - \sum_{i=0}^{r-1}\epsilon_j + r = \frac{r}{r+1}(n-\frakd) + \frac{5-r}{2}.
\]
Note that the upper bound $d_{\opt}$ for the distance of this code is
\begin{align*}
n - k - \left\lceil\frac{k}{r}\right\rceil + 2 &= n - \frac{r}{r+1}(n-\frakd) + \frac{r-5}{2} - \left\lceil\frac{1}{r+1}(n-\frakd) + \frac{5-r}{2r}\right\rceil + 2\\
&= \frakd + \frac{r-5}{2} + 2.
\end{align*}

The last equality holds for $r\ge 5$ whereas, for $r=3$, we just get $\frakd$.
So {the codes constructed this way} are optimal for $r=3$; for $r > 3$, these codes are further 
from the optimal bound the larger $r$ gets.

For $r$ even,  a similar calculation gives $\frakd+r/2$ as the upper bound for the distance when
$r > 2$ and $\frakd$ when $r=2$. 
So {the codes constructed} are optimal for $r=2$; for $r > 3$, these codes are further 
from the optimal bound the larger $r$ gets.

We note again the similarity with the Tamo-Barg 
codes discussed above, which uses a space of functions of the same form as
\eqref{fn} but with $\deg a_j \le k/r -1$
and a curve of the form $g(x)=t$ for a polynomial $g(x)$ in place of $C$. 
The length of their codes is at most $q$, whereas the codes above can be 
longer if the curve $C$ in~\eqref{crv} has more than $q$ affine points.

\subsection{The construction of Munuera and Ten\'orio}
\label{mt}

We briefly describe the general construction of \cite[Section 2.2]{MT}.
Here $t$ (to keep the notation of \cite{MT}) is a positive integer and not a variable as
elsewhere in this section.
They consider a map $(\phi_1,\ldots,\phi_t): \Aff^m \to \Aff^t$ and
another function $\phi_{t+1}:\Aff^m \to \Aff^1$. Their evaluation points lie in
$\Aff^m$ but they use the map
$ (\phi_1,\phi_2,\ldots,\phi_{t+1})$ to view them in $\Aff^{t+1}$ and for the purpose of
comparison
it is enough to consider $\Aff^{t+1} = \Aff^t \times \Aff^1$ and the
natural
projections  $\Aff^{t+1} \to \Aff^t$ and $\Aff^{t+1} \to \Aff^1$.
Their
$\Aff^1$ coordinate plays the role of the $\Aff^1_x$ coordinate in the
previous section. In particular, they use the properties of the rational
normal curve (under the guise of Lagrange interpolation) to get local recoverability. Their $\Aff^t$ plays the role of what we denote by
$\Aff^1_t$ in the previous section.

When it comes to explicit constructions they consider an algebraic curve
mapping to $\Aff^{t+1}$ (so the $\phi_i$ are functions on the curve)
and take the evaluation points from the image of
the curve. Their computation of the other parameters of the codes they
construct use the intrinsic geometry of the curve and not the geometry of
the curve within the ambient space, which is the viewpoint we will take in
Section \ref{s:P1xP1}. This is where our construction and theirs diverge.
Moreover, their examples lead to different code parameters which are not
directly comparable to ours. Particularly, they mostly deal with values of
the locality $r$ different from those that we consider. In \cite[Section 3.2, 3.3]{MT} they construct codes with $r=2$
and \cite[Section 3.4]{MT} they have codes with $r=q-1$ over $\F_{q^2}$. Whereas 
we, in Theorems \ref{prop:codesP1P1} and \ref{prop:refinedhirzebruch}, deal with $r$ 
such that $(r+1) \mid (q-1)$ over $\F_q$ and, in Theorem \ref{ell-thm}, with codes with
$r=q$ over $\F_{q^2}$. The one place where these intersect is the special case of $r=2$ 
in Theorem \ref{prop:codesP1P1} where we
deal with an elliptic curve inside our surface. These codes are then very similar to those of \cite[Section 3.3]{MT}
that also use an elliptic curve.
The ideas of \cite{MT} have been extended in
\cite{GHM} to construct $(r,\delta)$-LRC codes, which is a direction we do
not pursue here.

\section{Codes on ruled surfaces: $\PP^1\times \PP^1$}
\label{s:P1xP1}

{In this section we add one more layer of geometry to the codes we constructed in \S\ref{s:RuledSurfacesAffine} by considering codes on the ruled surface $S = \PP^1\times \PP^1$, which is a projective compactification of the surface $\Aff^1_x\times \Aff^1_t$. This extra layer of geometry affords important conceptual insights: a lower bound for the minimum distance of a code can be interpreted as an intersection number of two curves in $S$, and good lower bounds for a minimum distance can be achieved by forcing curves to intersect with high multiplicity at the point {$(\infty,\infty) \in S$}.}

{We begin with a toy model for our code, that is far from optimal, but which helps set ideas and notation.
We let $S := \PP^1_{(x:y)} \times \PP^1_{(t:u)}$, where $(x:y)$ and $(t:u)$ are respective homogeneous coordinates for the factors of $S$. }

\subsection{A coarse construction} 
\label{subsec:coarsecode}
Let $r$ be a positive odd integer, let $b \leq q$ be a positive integer, and set $n = b(r+1)$. Choose an integer $\frakd$ {divisible by $r+1$, so} that 
\[
N := \frac{n-\frakd}{r+1}
\]
is an integer, as well as {a positive integer $\alpha$}. Consider a curve of the form
\[
C : g(x,y;t,u) = 0
\]
in $S$, where $g$ is a bi-homogeneous polynomial of the bi-degree {$(r+1,\alpha)$}. In other words, every monomial of $g$ has total degree {$r+1$ in the variables $x$ and $y$, and total degree $\alpha$ in the variables $t$ and $u$}. We say that $C$ is of \defi{type} $(r+1,\alpha)$. Our code will be an evaluation code on the $\F_q$-vector space of functions of the form
{
\begin{equation}
\label{eq:fullsections}
\sigma = a_0(t,u)y^{r-1} + a_1(t,u)y^{r-2}x + \cdots a_{r-1}(t,u)x^{r-1},
\end{equation}
where the $a_i(t,u)$ are homogeneous polynomials of degree $N$ in $t$ and $u$. We write $V_{r-1,N}$} for this vector space. Each function {$\sigma \in V_{r-1,N}$} defines itself a curve in $X$ given by $\sigma = 0$. We write $(\sigma)$ for this curve\footnote{The notation $(\sigma)$ is the usual notation in algebraic geometry for the divisor of zeroes of a global section of a line bundle; see~\S\ref{subsec:AGcodes}.}; it is a curve of type $(r-1,N)$.

Write $p\colon S\to \PP^1_{(t:u)}$ for the projection onto the {second} factor. To construct our code, we pick $b$ points $(t_i:u_i) \in \PP^1_{(t:u)}(\F_q)$ such that the fiber $p^{-1}((t_i,u_i))\cap C$ consists of $r+1$ distinct points 
\[
(x_{i,1}:y_{i,1}),\dots,(x_{i,r+1}:y_{i,r+1}) \in \PP^1_{(x:y)}(\F_q)
\]
{
and set
\[
P_{i,j} = \left((x_{i,j}:y_{i,j}),(t_i:u_i) \right) \in S(\F_q).
\]
}
\begin{proposition}
\label{prop:coarseprojective}
The code
\[
\calC := \{ \left(\sigma(P_{i,j})\right)_{1\leq i \leq b, 1\leq j \leq r+1} : \sigma \in {V_{r-1,N}} \}
\]
has parameters satisfying
\begin{align*}
n &= b(r+1)\\
k &= r(N+1) = \frac{r}{r+1}\cdot(n-\frakd) + r \\
d &\leq \frakd -r+1 \\
d &\geq \frakd - \alpha(r-1)
\end{align*}
\end{proposition}

\begin{proof}
The parameter $k$ is simply the dimension of the $\F_q$-vector space {$V_{r-1,N}$}. The upper bound for the distance is the bound~\eqref{dist-bound}:
\begin{align*}
d 	&\leq n - k - \left\lceil\frac{k}{r}\right\rceil + 2 \\ 
	&= n - \frac{r}{r+1}\cdot(n-\frakd) - r - \frac{1}{r+1}\cdot(n-\frakd) - 1 + 2\\
	&= \frakd -r+1.
\end{align*}
We have used here the divisibility relation $(r+1) \mid (n - \frakd)$. For the lower bound on the distance, we note that the largest number of zeros in a code word in $\calC$ is bounded above by 
\[
\max_{\sigma \in V_{r-1,N}} \#\left(C\cap (\sigma)\right),
\]
i.e., the largest number of intersection points between $C$ and the curve $(\sigma) \subset S$ given by $\sigma = 0$, as $\sigma$ varies over the vector space {$V_{r-1,N}$}. The intersection theory of $S$ shows that this number is independent of $\sigma$: indeed, the intersection of divisors on $S$ of type $(a,b)$ and $(a',b')$ is $ab' + a'b$~\cite[V, Example~1.4.3]{Hartshorne}. Since $C$ is a curve of type {$(r+1,\alpha)$ and $(\sigma)$ is a curve of type $(r-1,N)$, we have
\begin{align*}
\#\left(C\cap (\sigma)\right) &= N(r+1) + \alpha(r-1) \\
&=n - \frakd + \alpha(r-1).
\end{align*}
Hence, the lowest weight for a code word in $\calC$ is
\[
d \geq n - \#\left(C\cap (\sigma)\right) = \frakd - \alpha(r-1),
\]
as claimed.}
\end{proof}

\begin{remark}
The codes in the above proposition have locality $r$.  However, we defer the discussion of locality until after we refine the code in the next section.
\end{remark}

\begin{remark}
The upper and lower bounds for $d$ in Proposition~\ref{prop:coarseprojective} meet if and only if $\alpha = 1$; this is precisely the habitat for the Tamo--Barg codes. In the notation of \S\ref{subsec:Tamo-Barg}, the affine curve $g(x) = t$ lies in the open set $\Aff^1_x\times \Aff^1_t = \{y,u\neq 0\}$ of $S$; its  projective closure in $S$ is given by {$y^{r+1}g(x/y)u = ty^{r+1}$}, which is a curve of type $(r+1,1)$ in the notation of this section.
\end{remark}

{
\begin{remark}
Let us compare the parameters in Proposition~\ref{prop:coarseprojective} with those of a base-line codes in Proposition~\ref{prop:paramsBaselineCode}. The length $n$, dimension $k$, and upper bound for $d$ coincide since we have specialized to the case where $M = N$ in Proposition~\ref{prop:coarseprojective}. If $r \geq 3$, then the lower bound for $d$ in Proposition~\ref{prop:paramsBaselineCode} is $2(b - N)$, while the bound for the codes just studied is
\[
d \geq (r+1)(b - N) - \alpha(r-1)
\]
The latter bound is better as long as $b > N + \alpha$, i.e., as long as $\frakd > \alpha(r+1)$.
\end{remark}
}

\subsection{Refining the construction}

In this section, we show that one can narrow the gap between the upper and lower bounds for $d$ in Proposition~\ref{prop:coarseprojective} {by 
\begin{enumerate}
\item choosing $C$ judiciously, 
\smallskip
\item using a particular proper subspace $V\subset V_{r-1,N}$ for the evaluation code, 
\smallskip
\item using only points $P_{i,j} = \left((x_{i,j}:y_{i,j}),(t_i:u_i)\right)$ with {$y_{i,j} = u_i = 1$}. 
\end{enumerate}
}
Intuitively, our construction guarantees that the point
\[
{(\infty,\infty)} := \left((1:0),(1:0)\right) \in S(\F_q)
\]
lies in the intersection $C \cap (\sigma)$ for all $\sigma \in V$ with high multiplicity.  {Note that $P_{i,j} \neq (\infty,\infty)$ for $i$ and $j$ by construction of $P_{i,j}$.}  This allows us certify the code $\calC$ has minimum distance $d = \frakd$.

Consider the curve
{
\[
C : u^\alpha x^{r+1} - (t^\alpha + u^\alpha)y^{r+1} = 0,
\]
which is a particular curve of type $(r+1,\alpha)$ in $S$.} We shall use functions of the form~\eqref{eq:fullsections}, but we constrain the degree in $t$ of the polynomials $a_i(t,u)$, as follows:
\[
\deg_t a_i(t,1) \leq N - \left\lceil\frac{\alpha i}{r+1}\right\rceil.
\]
{This requires $N \ge \left\lceil{\alpha (r-1)}/(r+1)\right\rceil$, which we now assume}. In other words, setting
\[
\epsilon_i := \left\lceil\frac{\alpha i}{r+1}\right\rceil,
\]
we assume that for each $0\leq i \leq r-1$,
\[
a_i(t,u) = u^{\epsilon_i} \cdot a_i'(t,u)
\]
for a homogeneous polynomial $a_i'(t,u)$. When this is the case, the vector space of functions
{
\[
V := \{ \sigma \in V_{r-1,N} : \sigma = a_0(t,u)y^{r-1} + u^{\epsilon_1}\cdot a_1(t,u)y^{r-2}x + \cdots + u^{\epsilon_{r-1}} a_{r-1}(t,u)x^{r-1} \}
\]
}
has dimension
\begin{equation}
\label{eq:dimGoodCode}
k = r(N+1) - \sum_{i=0}^{r-1}\epsilon_i.
\end{equation}
{The vector space $V$ has the important property that $\sigma((\infty,\infty)) = 0$ for all $\sigma\in V$. This is key in improving our bounds for the minimum distance of the codes we define using the curve $C$ and the space of functions $V$.} We pick $b$ points $(t_i:1) \in \PP^1_{(t:u)}(\F_q)$ such that the fiber $p^{-1}((t_i:1))\cap C$ consists of $r+1$ distinct points 
\[
(x_{i,1}:1),\dots,(x_{i,r+1}:1) \in \PP^1_{(x:y)}(\F_q).
\]
Put
\[
P_{i,j} = \left( (x_{i,j}:1),(t_i:1) \right) \in S(\F_q).
\]
{
\begin{theorem}
\label{prop:codesP1P1}
Assume that $\alpha \mid (r+1)$ and $(r+1) \mid (q-1)$. The code
\[
\calC := \{ \left(\sigma(P_{i,j})\right)_{1\leq i \leq b, 1\leq j \leq r+1} : \sigma \in V \}
\]
has locality $r$ and its parameters satisfy
\begin{align*}
n &= b(r+1),\\
k &=
\begin{cases}
r(N+1) - \frac{r(r-1)}{2} \, , \text{if } r+1= \alpha,  \, \text{and} \\
r(N+1) + 2\alpha - \frac{(\alpha + 1)(r+1)}{2} \, , \text{if } r+1 > \alpha,
\end{cases}\\
d &\leq \frakd + \frac{(\alpha - 1)(r-3)}{2} - \left\lceil \frac{2\alpha}{r} - \frac{(\alpha + 1)(r+1)}{2r}\right\rceil, \\
d &\geq \frakd.
\end{align*}
In particular, the code $\calC$ is an optimal LR code if $\alpha = 1$ or $r=3$.
\end{theorem}
}

{
\begin{example}
\label{example-4q}
Setting $\alpha = 2$, $r = 3$, and picking an integer $d$ divisible by $4$ such that $4b\geq d$, we obtain optimal LR codes with parameters
\[
(n,k,d,r) = \left(4b, 3b - \frac{3}{4}d + 1, d, 3\right).
\]
Since $b \leq q$, one can construct codes with $n = 4q$ with high information rate that are locally recoverable.  Compare this with the baseline codes from Example~\ref{ex:goodcode}, where a code with similar parameters is possible only when $d = 4$.
\end{example}
}

\begin{proof}[Proof of Theorem~\ref{prop:codesP1P1}]
Assume that $r+1 > \alpha$.  By~\eqref{eq:dimGoodCode}, to establish the claim on $k = \dim_{\F_q} V$, it suffices to show that
\[
\sum_{i=0}^{r-1}\epsilon_i = \frac{(\alpha + 1)(r+1)}{2} - 2\alpha.
\]
The sequence of integers $\epsilon_0,\dots,\epsilon_{r-1}$ has the form
\[
0,\underbrace{1,\dots,1}_{(r+1)/\alpha}, \underbrace{2,\dots,2}_{(r+1)/\alpha}, \underbrace{3,\dots,3}_{(r+1)/\alpha}, \dots, \underbrace{\alpha - 1,\dots,\alpha - 1}_{(r+1)/\alpha}, \underbrace{\alpha,\dots,\alpha}_{(r+1)/\alpha - 2}.
\]
Hence
\begin{align*}
\sum_{i=0}^{r-1}\epsilon_i &= \sum_{l = 1}^{\alpha - 1} l\cdot\frac{r+1}{\alpha} + \alpha\left(\frac{r+1}{\alpha} - 2\right) \\
&= \frac{(\alpha - 1)\alpha}{2}\cdot \frac{r+1}{\alpha} + (r+1) - 2\alpha \\
&= (\alpha -1)\frac{r+1}{2} + (r+1) - 2\alpha \\
&= \frac{(\alpha + 1)(r+1)}{2} - 2\alpha.
\end{align*}

If $r+1 = \alpha$, then $\epsilon_i = i$ and the result follows.

For the lower bound on the distance, note that the largest number of zeros in a code word in $\calC$ is bounded above by 
\[
\max_{\sigma \in V} \#\left(C\cap (\sigma)\right),
\]
just as in Proposition~\ref{prop:coarseprojective}. We have already seen that
\[
C\cdot (\sigma) = \alpha(r-1) + n - \frakd.
\]
However, for every $\sigma \in V$, the curves $C$ and $(\sigma)$ intersect at {the point $(\infty,\infty) \in S(\F_q)$}. We claim this happens with multiplicity \emph{at least} $\alpha(r-1)$, and hence
\[
\max_{\sigma \in V} \#\left(C\cap (\sigma)\right) \leq C\cdot (\sigma) - \alpha(r-1) = n - \frakd,
\]
from which we deduce that
\[
d \geq n - \max_{\sigma \in V} \#\left(C\cap (\sigma)\right) \geq \frakd.
\]
To establish the claim on the multiplicity of $C$ and $(\sigma)$ at {$(\infty,\infty)$}, note that the point {$(\infty,\infty)$} is the origin of the affine patch {$\Aff^2_{(y,u)}$ of $S$}. In this patch, an affine equation for $C$ is
\[
C : u^\alpha = (1 + u^\alpha)y^{r+1},
\]
which is in fact singular at the origin (this only helps increase the multiplicity of the intersection with the curve $(\sigma)$). In the complete local ring of $C$ at the origin, the quantity $1 + u^\alpha$ has an $\alpha$-th root. More precisely, let
\[
{A = k[y,u]/(u^\alpha - (1 + u^\alpha)y^{r+1})}
\]
be the affine coordinate ring of $C$, and let {$\frakm = (y,u)$} be the maximal ideal corresponding to the origin. Then in the completed local ring $\hat{A}_\frakm$, the binomial expansion shows that
\[
w = (1 + u^\alpha)^{1/\alpha} = 1 + \binom{1/\alpha}{1}u^\alpha + \binom{1/\alpha}{2}u^{2\alpha} + \binom{1/\alpha}{3}u^{3\alpha} + \cdots
\]
Let $\zeta$ denote an $\alpha$-th root of unity in an algebraic closure of $\F_q$. Geometrically, $C$ has $\alpha$ branches at the origin:
\[
{u = wy^{(r+1)/\alpha},u = \zeta wy^{(r+1)/\alpha},\dots, u = \zeta^{\alpha -1}wy^{(r+1)/\alpha},}
\]
For each one of these branches, $y$ is a uniformizer for the ideal $\frakm$, and $u$ has valuation $(r+1)/\alpha$ with respect to this uniformizer\footnote{By this we mean: let {$B = \bar\F_q[y,u]/(u - \zeta^i wy^{(r+1)/\alpha})$} be the geometric local coordinate ring of one of the branches of $C$. Then the $\frakm$-adic completion $\hat B_{\frakm}$ at the maximal ideal {$\frakm = (y,u)$} corresponding to the origin is a local discrete valuation ring. Hence the ideal $\frakm\hat B_{\frakm}$ is principal~\cite[Proposition~9.2]{AtiyahMacdonald}. The equation of the branch shows that $y$ is a generator for this ideal, and that $u \in \frakm^{(r+1)/\alpha}\setminus\frakm^{(r+1)/\alpha - 1}$, which is to say that $u$ has $\frakm$-adic valuation $(r+1)/\alpha$.}.  For $\sigma \in V$, a local equation for $(\sigma)$ in the affine patch {$\Aff^2_{(y,u)}$} is
\[
{a_0(1,u)y^{r-1} + u^{\epsilon_1}\cdot a_1(1,u)y^{r-2} + \cdots + u^{\epsilon_{r-1}} a_{r-1}(1,u) = 0}
\]
The monomial $u^{\epsilon_i}y^{r-1-i}$ has $\frakm$-adic valuation
\[
\left\lceil\frac{\alpha i}{r+1}\right\rceil\cdot\frac{r+1}{\alpha} + r - 1 - i.
\]
As $i$ ranges through $0,\dots,r-1$, the \emph{smallest} value of this quantity is $r-1$. Hence, on each branch of $C$ the minimal $\frakm$-adic valuation of $\sigma \in V$ is $r-1$, and therefore $C$ and $(\sigma)$ intersect at {$(\infty,\infty)$} with multiplicity $\geq \alpha(r-1)$.  This concludes the proof of the lower bound for $d$.

Next, we compute {an upper bound for $d$ using~\eqref{dist-bound}}:
\begin{align*}
d 			&\leq n - k - \left\lceil\frac{k}{r}\right\rceil + 2\\
			&= n - r(N+1) - 2\alpha + \frac{(\alpha + 1)(r+1)}{2} - (N+1) - \left\lceil\frac{2\alpha}{r} - \frac{(\alpha + 1)(r+1)}{2r}\right\rceil + 2\\
			&= \frakd - (r+1) - 2\alpha + \frac{(\alpha + 1)(r+1)}{2} - \left\lceil\frac{2\alpha}{r} - \frac{(\alpha + 1)(r+1)}{2r}\right\rceil + 2\\
			&= \frakd + \frac{(\alpha - 1)(r - 3)}{2} - \left\lceil\frac{2\alpha}{r} - \frac{(\alpha + 1)(r+1)}{2r}\right\rceil.
\end{align*}

{
Finally, we discuss the locality of the code $\calC$. Since all points $P_{i,j}$ used to construct $\calC$ have {$y_{i,j} = u_i = 1$}, the set $\{P_{i,j}\}$ lies entirely in the affine patch $\Aff_x^1\times\Aff^1_t$ of $S$. Proceeding as in \S\ref{ss:ruledperspective}, we map this affine patch to $\Aff^{r-1}\times \Aff_t$ via 
\[
(x,t) \mapsto (x,x^2,\dots,x^{r-1};t).
\]
The image of the points $\{P_{i,j}\}$ lie on a rational normal curve, so no $r$ of them lie on a hyperplane, and hence Lemma~\ref{lem:BaselineRecoverability} shows the code $\calC$ has locality $r$.
}
\end{proof}

{
\begin{remark}
Let us compare the parameters in Proposition~\ref{prop:codesP1P1} with those of a base-line codes in Proposition~\ref{prop:paramsBaselineCode}. The length $n$ is the same for both constructions.  The dimension is smaller in Proposition~\ref{prop:codesP1P1}; however, on the one hand, $\frakd = (b - N)(r+1)$, and on the other hand, when $M = N$ and $r \geq 3$ the lower bound for the distance in the base-line codes is $2(b - N)$. Hence, the lower bound $\frakd$ represents an improvement on base-line codes of $(r-1)(b - N)$. For a numerical example, take $r+1=\alpha=5$ and $q=16$.
Then we can take $b=10$, so $n=50$ and $\frakd$ can be any integer divisible by $5$ with $\frakd \le 35$ and the parameters are given as in Proposition~\ref{prop:codesP1P1} with equality $d = \frakd$.
\end{remark}
}

\section{Codes on Hirzebruch surfaces}

The ruled surface $\PP^1\times \PP^1$ is an example of a Hirzebruch surface, which are ruled surfaces determined by a non-negative integer $m$. After recalling some of the geometry of these surfaces, we adapt the construction of codes in \S\ref{s:P1xP1} to the setting of Hirzebruch surfaces.

\subsection{Hirzebruch surfaces {$\mathbb{F}(m)$}}

Let $m \in \Z_{\geq 0}$; we let two copies of the multiplicative group $\G_m\times \G_m$ act on the product of two punctured affine planes $\mathbb{A}^2\setminus \{ (0,0)\}\times \mathbb{A}^2\setminus \{(0,0)\}$ via
\begin{align*}
 (\lambda, 1): (x,y;t,u) &\mapsto (\lambda^{-m} x, y ; \lambda t, \lambda u)\\
 (1,\mu): (x,y;t,u) &\mapsto  (\mu x,\mu y;t,u). 
 \end{align*}
The \defi{Hirzebruch surface} {$S = \F(m)$} is the quotient
\[
\mathbb{A}^2\setminus \{ (0,0)\}\times \mathbb{A}^2\setminus \{(0,0)\}\ / \mathbb{G}_m \times \mathbb{G}_m.
\]
Such surfaces are endowed with a natural fibration $p\colon S \to \mathbb{P}^1_{(t:u)}$ given by 
\begin{align}\label{hirzebruch_fibration}
 \left( (x:y) , (t:u) \right) \to (t:u).
\end{align}
{Note that $\PP^1\times \PP^1 = \F(0)$.}

\begin{lemma}\label{basicshirzebruch} Let {$S = \F(m)$} be as above. The following hold:
\begin{enumerate}
\item The Picard group $\mathrm{Pic}(S)$ is isomorphic to $\Z^2$, generated by the classes of the curves 
\[
A=\{t-u=0\} \quad\text{and}\quad B=\{x=0\},
\]
which are, respectively, a fiber of \eqref{hirzebruch_fibration} and the so-called negative section of $S$.
\medskip
\item The intersection pairing on $\Pic(S)$ is determined by 
\[
A^2=0, \quad A\cdot B=1 \quad\text{and}\quad B^2= {-m}.
\]
\item Let $M = mA + B \in \Pic(S)$. The canonical divisor $K_S$ is linearly equivalent to $(m-2)A-2M$.
\medskip
\item For non-negative integers $\alpha$, $\beta$ satisfying $\alpha\geq m\beta -1$, the Riemann--Roch space 
$\calL(S,\alpha A + \beta B)$ has dimension
\[
\ell(S,\alpha A + \beta B) = (\alpha+1)(\beta+1)- m\frac{\beta(\beta +1)}{2}.
\]
\end{enumerate}
\end{lemma}

\begin{proof}
For (1), (2) and (3) see \cite{Reid}, Sections B.2.9 and B 2.7.  
The Riemann--Roch theorem for surfaces gives the Euler characteristic of the class $\alpha A + \beta B$:
\[
\frac{(\alpha A + \beta B)\cdot(\alpha A + \beta B - K_S)}{2} + 1 = (\alpha+1)(\beta+1)- n\frac{\beta(\beta +1)}{2}.
\]
By, e.g.,~\cite[Thm. 2.1.]{Coz}, the conditions $\beta \geq 0$ and $\alpha\geq m\beta -1$ guarantee that this Euler characteristic coincides with the dimension of the Riemann--Roch space $\calL(S,\alpha A + \beta B)$.
\end{proof}

\begin{remark}
The morphism $\phi\colon X\rightarrow  \bar{X} \subset \mathbb{P}^m$ defined by the sections generating the projectivized Riemann-Roch space $|M|$ is the natural resolution of the cone over the rational normal curve of degree $n$. The map $\phi$ contracts $B$ to the vertex of the cone (see~\cite[B~2.9]{Reid}).
\end{remark}

\subsection{Riemann-Roch spaces for codes}
In this section, we give an explicit description of the elements of the Riemann--Roch spaces $V_{\beta,\alpha} := \calL(S, \alpha A + \beta B)$ appearing in Lemma~\ref{basicshirzebruch}. We assume throughout that $\alpha$ and $\beta$ are non-negative integers.

\begin{lemma}
 Let $\alpha=\varepsilon +m\beta$ with $\varepsilon \geq 0$.  The elements of $V_{\beta,\alpha}$ have the form
\begin{equation}
\label{hirzebruch:vectorspace}
\sigma= a_0(t,u)y^{\beta}+ a_1(t,u)y^{\beta-1}x + \cdots + a_{\beta}(t,u)x^{\beta}
\end{equation}
where $a_i(t,u)$ is a homogeneous polynomial of degree $\varepsilon + im$ for $i = 0,\dots,\beta$. We have
\[
\dim V_{\beta,\alpha} = (\alpha + 1)(\beta + 1) - m\frac{\beta(\beta + 1)}{2}.
\]
\end{lemma}
\begin{proof}
 Let  $\sigma$ be as in the statement of the lemma. First, we show that $\sigma \in V_{\beta,\alpha}$. Since $A$ and $B$ generate $\mathrm{Pic}(S)$, there are $\alpha'$ and $\beta'$ such that $(\sigma) = \alpha' A + \beta' B$ as classes in $\Pic(S)$. To determine $\alpha'$ and $\beta'$ we use the intersection pairing on $\Pic(S)$.
 
 Since $A$ is a curve defined by fixing the ratio $t/u$, we have that 
 \[
  (\sigma)\cdot A= \beta.
 \]
On the other, since $B= \{x=0\}$, we see that
\[
 (\sigma)\cdot B= \varepsilon .
\]
We obtain the system of equations
\begin{align*}
 \beta &=  (\sigma)\cdot A = \alpha'\cdot A^2+ \beta' A\cdot B= \beta',
 \\
 \varepsilon &= (\sigma)\cdot B= \alpha' A\cdot B + \beta' B^2= \alpha'-m\beta'.
\end{align*}
Thus $\beta' = \beta$ and $\alpha' = \varepsilon + m \beta' = \alpha$ as claimed. Note that the condition that $a_i(t,u)$ is homogeneous of degree $\varepsilon +im$ ensures that the monomials are invariant under the action $(\lambda, 1) \in \mathbb{G}_m \times \mathbb{G}_m$.

The subspace of $V_{\beta,\alpha}$ generated by elements of the form~\eqref{hirzebruch:vectorspace} has dimension 
\begin{align*}
 k &= (\varepsilon +1)+ (\varepsilon+ 1+m) + \cdots +(\varepsilon +1+\beta m) \\
 &= \sum_{i=0}^{\beta} (\varepsilon+1)+im \\
&=(\beta+1)(\varepsilon+1)+m\frac{\beta(\beta+1)}{2}\\
&= (\alpha + 1)(\beta + 1) - m\frac{\beta(\beta + 1)}{2}
\end{align*}
and hence must be equal to the entire vector space, by Lemma~\ref{basicshirzebruch}(4).
\end{proof}

{
\subsection{A coarse construction}

Let $r$ and $b \leq q$ be positive integers, and set $n = b(r+1)$. Choose an integer $\frakd$, divisible by $r+1$, so that 
\[
N := \frac{n-\frakd}{r+1}
\]
is an integer, as well as a positive integer $\alpha$. Set $\beta = r - 1$, and consider a curve of the form
\[
C : g(x,y;t,u) = 0
\]
in $S$, where $g$ is an element of 
\[
V_{r+1,\alpha + m(r+1)} = \calL\left( (\alpha + m(r+1))A + (r+1)B\right).
\]
We say $C$ is of \defi{type} $(r+1,\alpha + m(r+1))$.} The fibration $p\colon S \to \PP^1_{(t:u)}$ in~\eqref{hirzebruch_fibration} gives $S$ the structure of a ruled surface.
To construct evaluation codes using {$C$}, pick $b$ points $(t_i:u_i) \in \PP^1_{(t:u)}(\F_q)$ such that the fiber $p^{-1}((t_i:u_i))\cap {C}$ consists of $r+1$ distinct points 
\[
(x_{i,1}:y_{i,1}),\dots,(x_{i,r+1}:y_{i,r+1}).
\]
Put
\[
P_{i,j} = \left( (x_{i,j}:y_{i,j}),(t_i:u_i) \right) \in S(\F_q),
\]
so that there are $n = b(r+1)$ points of the form $P_{i,j}$ in total. We shall use the vector space 
\[
V_{\beta,N + m\beta} = V_{r-1,N + m(r-1)}
\]
to construct our evaluation codes.

\begin{proposition}
\label{prop:coarsehirzebruch}
The code
\[
\calC := \{ \left(\sigma(P_{i,j})\right)_{1\leq i \leq b, 1\leq j \leq r+1} : \sigma \in V_{r-1,N + m(r-1)} \},
\]
constructed using {$C$}, has locality $r$ and its parameters satisfy
\begin{align*}
n &= b(r+1)\\
k & = (N+1)r + m\frac{r(r-1)}{2}\\
d &\leq  \frakd- (r-1) - m\frac{(r^2-1)}{2}\\
d & \geq \frakd- (r-1)(\alpha+m(r+1)).
\end{align*}
\end{proposition}

\begin{proof}
By Lemma~\ref{hirzebruch:vectorspace}, we have
\begin{equation}
\label{eq:dimRRforHirzebruch}
k = \dim V_{r-1,N + m(r-1)} = r(N+1)+m\frac{r(r-1)}{2}.
\end{equation}
Next, if $r$ is odd or $m$ even, we have
\[
 \left\lceil\frac{k}{r}\right\rceil= N+1 +m\frac{(r-1)}{2}.
\]
Otherwise, 
\[
 \left\lceil\frac{k}{r}\right\rceil= N+1 +m\frac{(r-1)}{2} +\frac{1}{2} \ge N+1 +m\frac{(r-1)}{2}.
\]
Hence, an upper bound for $d$ using~\eqref{dist-bound} is
\begin{align*}
 d &\leq n-k-\left\lceil\frac{k}{r}\right\rceil + 2 \\
 &\le n-r(N+1)-m\frac{r(r-1)}{2} - (N+1)-m\frac{(r-1)}{2}+2\\
 &=n-(n-\frakd) - (r+1) -m\frac{r(r-1)}{2} -m\frac{(r-1)}{2}+2\\
 &= \frakd-(r-1)-m\frac{(r^2-1)}{2}.
\end{align*}
As in the proof of Proposition~\ref{prop:codesP1P1}, a lower bound for the minimum distance of $\calC$ is
\begin{align*}
d &\geq n - \max_{\sigma \in V} \#\left({C}\cap (\sigma)\right)\\
&\geq n - {C}\cdot (\sigma) \text{ for any }\sigma \in V_{r-1,N + m(r-1)}
\end{align*}
Since the equation for {$C$} is an element of $V_{r+1,\alpha + m(r+1)}$, we may use Lemma~\ref{basicshirzebruch}(2) to compute
\begin{align*}
{C}\cdot (\sigma) &= \left( (\alpha + m(r+1))A + (r+1)B \right)\cdot \left( (N + m(r-1))A + (r-1)B \right)\\
&= (r-1)(\alpha + m(r+1)) + (N + m(r-1))(r+1) - m(r^2-1) \\
&= (r-1)(\alpha + m(r+1)) + n - \frakd,
\end{align*}
and hence
\[
d \geq \frakd - (r-1)(\alpha + m(r+1)).
\]
as claimed. 
Finally, the locality is $r$ by the same argument as in the end of the proof of Proposition \ref{prop:codesP1P1}. 
\end{proof}

\begin{remark}
When $m=0$, we have {$S = \mathbb{F}(0)= \mathbb{P}^1\times\mathbb{P}^1$}.  In this case, the bounds on the distance for $\calC$ coincide with the bounds of Proposition~\ref{prop:coarseprojective}, as one would expect.
\end{remark}

\begin{remark}
The upper and lower bounds for the minimum distance in Proposition~\ref{prop:coarsehirzebruch} meet when
\[
1 + m\frac{(r+1)}{2} = \alpha + m(r+1).
\]
Since $\alpha$, $m$ and $r$ are non-negative, we must have $m=0$ (i.e., $S = \PP^1\times \PP^1$) and $\alpha =1$.
\end{remark}

{
\subsection{Refining the Construction}
Consider the curve $C\subset S$ with affine model given by
\begin{equation*}
C  : x^{r+1} = t^{\alpha} + 1.
\end{equation*}
The projective closure of this curve in $S$ is given by:
\begin{equation}
 u^{\alpha + m(r+1)}x^{r+1} - (t^{\alpha}+u^{\alpha})y^{r+1} = 0.
\end{equation}
The left hand side of the above equation is an element of the vector space $V_{r+1,\alpha + m(r+1)}$.

To construct evaluation codes using {$C$}, as usual, pick $b$ points $(t_i:u_i) \in \PP^1_{(t:u)}(\F_q)$ such that the fiber $p^{-1}((t_i:u_i))\cap {C}$ consists of $r+1$ distinct points 
\[
(x_{i,1}:y_{i,1}),\dots,(x_{i,r+1}:y_{i,r+1}).
\]
Put
\[
P_{i,j} = \left( (x_{i,j}:y_{i,j}),(t_i:u_i) \right) \in S(\F_q),
\]
so that there are $n = b(r+1)$ points of the form $P_{i,j}$ in total.  For the vector space of function on which we evaluate the $P_{i,j}$, we constrain the degree in $t$ of the polynomials $a_i(t,u)$, as follows:
\[
\deg_t a_i(t,1) \leq N + im - \left\lceil\frac{i(\alpha + m(r+1))}{r+1}\right\rceil.
\]
{Again, this requires $N \ge \left\lceil{\alpha (r-1)}/(r+1)\right\rceil$, which we now assume}.
In other words, setting
\[
\epsilon_i := \left\lceil\frac{i(\alpha  + m(r+1))}{r+1}\right\rceil,
\]
we assume that for each $0\leq i \leq r-1$,
\[
a_i(t,u) = u^{\epsilon_i} \cdot a_i'(t,u)
\]
for a homogeneous polynomial $a_i'(t,u)$. When this is the case, the calculation~\eqref{eq:dimRRforHirzebruch} shows that the vector space of functions
\[
V := \{ \sigma \in V_{r-1,N + m(r-1)} : \sigma = a_0(t,u)y^{r-1} + u^{\epsilon_1}\cdot a_1(t,u)y^{r-2}x + \cdots + u^{\epsilon_{r-1}} a_{r-1}(t,u)x^{r-1} \}
\]
has dimension
\[
k = r(N+1)+m\frac{r(r-1)}{2} - \sum_{i=0}^{r-1}\epsilon_i
\]
If $\alpha = r+1$ then 
\[
\sum_{i=0}^{r-1}\epsilon_i = \sum_{i=0}^{r-1}i + im = (m+1)\frac{r(r-1)}{2}
\]
Otherwise, if $r+1 > \alpha$ then 
\begin{align*}
\sum_{i=0}^{r-1}{\epsilon_i} &= \sum_{i = 0}^{r-1}\left\lceil\frac{i\alpha}{r+1}\right\rceil + im \\
&= \frac{(\alpha + 1)(r+1)}{2} - 2\alpha + m\frac{r(r-1)}{2}.
\end{align*}
where the second equality follows by our work in the proof of Proposition~\ref{prop:codesP1P1}.
We conclude that
\[
k = \begin{cases}
r(N+1) - \frac{r(r-1)}{2} \, , \text{if } r+1= \alpha,  \, \text{and} \\
r(N+1) + 2\alpha - \frac{(\alpha + 1)(r+1)}{2} - m\frac{r(r-1)}{2} \, , \text{if } r+1 > \alpha,
\end{cases}
\]

\begin{theorem}
\label{prop:refinedhirzebruch}
Assume that $\alpha \mid (r+1)$ and $(r+1) \mid (q-1)$. The code
\[
\calC := \{ \left(\sigma(P_{i,j})\right)_{1\leq i \leq b, 1\leq j \leq r+1} : \sigma \in V \}
\]
has locality $r$ and its parameters satisfy
\begin{align*}
n &= b(r+1),\\
k &=
\begin{cases}
r(N+1) - \frac{r(r-1)}{2} \, , \text{if } r+1= \alpha,  \, \text{and} \\
r(N+1) + 2\alpha - \frac{(\alpha + 1)(r+1)}{2} - m\frac{r(r-1)}{2} \, , \text{if } r+1 > \alpha,
\end{cases} \\
d &\leq \frakd + \frac{(\alpha - 1)(r-3)}{2} - \left\lceil \frac{2\alpha}{r} - \frac{(\alpha + 1)(r+1)}{2r}\right\rceil + m\frac{(r^2 - 1)}{2}, \\
d &\geq \frakd.
\end{align*}
\end{theorem}

\begin{proof}
We have already discussed the values of $n$ and $k$ above.  The upper bound for $d$ is obtained from~\eqref{dist-bound}, proceeding as in the proof of Proposition~\ref{prop:codesP1P1}.

For the lower bound on the distance, we note that, as before,
\[
d \leq \max_{\sigma \in V} \#\left(C\cap (\sigma)\right),
\]
just as in Proposition~\ref{prop:coarseprojective}. In the course of the proof of Proposition~\ref{prop:coarsehirzebruch}, we saw that
\[
C\cdot (\sigma) = (r-1)(\alpha + m(r+1)) + n - \frakd.
\]
However, for every $\sigma \in V$, the curves $C$ and $(\sigma)$ intersect at the point
\[
[x,y;t,u] = [1,0;1,0] \in \F(m).
\]
We claim this happens with multiplicity \emph{at least} $(r-1)(\alpha + m(r+1))$, and hence
\[
\max_{\sigma \in V} \#\left(C\cap (\sigma)\right) \leq C\cdot (\sigma) - (r-1)(\alpha + m(r+1)) = n - \frakd,
\]
from which we deduce that
\[
d \geq n - \max_{\sigma \in V} \#\left(C\cap (\sigma)\right) \geq \frakd.
\]
The claim on the multiplicity is established as in the proof of Proposition~\ref{prop:codesP1P1}: the point $[1,0;1,0] \in \F(m)$ is the origin of the affine patch of $C$ given by
\[
u^{\alpha + m(r+1)} = (1 + u^\alpha)y^{r+1},
\]
In the complete local ring of $C$ at the origin, the quantity $1 + u^\alpha$ has an $(\alpha + m(r+1))$-th root. Let $\zeta$ denote an $(\alpha + m(r+1))$-th root of unity in an algebraic closure of $\F_q$. Geometrically, $C$ has $\alpha + m(r+1)$ branches at the origin:
\[
{u = wy^{(r+1)/(\alpha + m(r+1))},u = \zeta wy^{(r+1)/(\alpha + m(r+1))},\dots, u = \zeta^{(\alpha + m(r+1)) -1}wy^{(r+1)/(\alpha + m(r+1))},}
\]
For each one of these branches, $y$ is a uniformizer for the maximal ideal at the origin of $C$, and $u$ has valuation $(r+1)/(\alpha + m(r+1))$ with respect to this uniformizer (see the proof of Proposition~\ref{prop:codesP1P1} for more details).  For $\sigma \in V$, a local equation for $(\sigma)$ in the affine patch {$\Aff^2_{(y,u)}$} is
\[
{a_0(1,u)y^{r-1} + u^{\epsilon_1}\cdot a_1(1,u)y^{r-2} + \cdots + u^{\epsilon_{r-1}} a_{r-1}(1,u) = 0}
\]
The monomial $u^{\epsilon_i}y^{r-1-i}$ has $\frakm$-adic valuation
\[
\left\lceil\frac{i(\alpha + m(r+1))}{r+1}\right\rceil\cdot\frac{r+1}{(\alpha + m(r+1))} + r - 1 - i.
\]
As $i$ ranges through $0,\dots,r-1$, the \emph{smallest} value of this quantity is $r-1$. Hence, on each branch of $C$ the minimal valuation at the origin of $\sigma \in V$ is $r-1$, and therefore $C$ and $(\sigma)$ intersect at $[1,0;1,0]$ with multiplicity $\geq \alpha(r-1)(\alpha + m(r+1))$.  This concludes the proof of the lower bound for $d$.
\end{proof}
}

When $m=0$, we recover Proposition~\ref{prop:codesP1P1}. The parameters get
slightly worse for $m>0$ but this more general construction might still be
interesting.

\section{Locally recoverable codes from elliptic surfaces}

\subsection{Elliptic surfaces}

The definitions of this section hold over an arbitrary field $k$. 

An algebraic surface $\calE$ is called an \defi{elliptic surface} if it is endowed with a morphism $\pi: \calE \to B$ to a base algebraic curve $B$ such that

\begin{enumerate}
 \item[i)] for all but finitely many $t \in B(\bar{k})$, the fiber $\pi^{-1}(t)$ is a genus one curve, where $\bar{k}$ is a fixed algebraic closure of $k$.
 \item[ii)] there is a section $\sigma$ to $\pi$, i.e., a morphism $\sigma: B \rightarrow \calE$ such that $ \pi \circ \sigma = \mathrm{id}_B$.
\end{enumerate}
                                                                                                                                                                  
The morphism $\pi$ is called an elliptic fibration. Condition ii) implies that all but finitely many fibers of $\pi$ are indeed elliptic curves.

	Let $\pi: \calE \to B$ be an elliptic fibration. A section $P: B \to \calE$ is, by definition, a regular map such that $\pi \circ P$ is the identity on $B$. We denote by $\calO$ the zero section and by abuse of notation
also the zero element of any fiber. The set of sections of the fibration $\pi$ in the above sense can be made into an
abelian group with identity $\calO$ (in the same way one defines the group law on an elliptic curve). This group is
called the Mordell-Weil group of $\calE$ and it is finitely generated by the N\'eron-Severi-Mordell-Weil theorem.

We also have that $ \calE$ has a Weierstrass equation

	\[
		y^2+a_1xy+a_3y = x^3+a_2x^2+a_4x+a_6
	\]
where $a_i \in k(B)$.  We consider the divisor $D = n\cdot\infty + m\cdot \calO$, where $\infty$ is the ``fiber above $\infty$'', and $\calO$ is the zero section. A function on $\calE$ whose polar divisor is bounded by $D$ is of the form
	\[
		\sum_{2i \leq m}\alpha_i x^i + \sum_{2i + 3 \leq m} \beta_i x^iy,
	\]
where $\alpha_i$ and $\beta_i$ are functions in the Riemann Roch space $\calL(B,n\cdot \infty)$.

Each fiber $E$ is embedded in $\PP^{n-1}$ by the linear system $|n\calO|$ (where
$\calO$ is the identity of $E$).

\subsection{General code construction}
\bigskip

Let $\pi: \calE \to B$ be an elliptic fibration. We denote by $\calO$ the zero section and by abuse of notation
also the zero element of any fiber. We denote by $E_t= \pi^{-1}(t)$ the fiber above $t$ and by $E_t[2]$ its subgroup
of elements of order at most $2$.

\begin{lemma} 
\label{ell-recover}
Assume that for each $t$ in a subset of $B(\F_q)$
such that the fiber $E_t$  over $t$ is an elliptic curve, 
we are given $\Gamma_t \subset E_t(\F_q)-E_t[2]$ all of same cardinality $r+1$ for some integer $r$
with the property that $\sum_{P \in \Gamma_t} P \in E_t[2]$ in the group law of $E_t$.

Let $\Gamma = \bigcup_t \Gamma_t$ and $V$ a finite-dimensional $\F_q$-vector space of functions on $\calE$
such that the restriction of any element of $V$ to a fiber above any $t$ is in the Riemann-Roch space $\calL(E_t,r\calO)$.
We form a code $\calC$ by evaluating the functions on $V$ on the points of $\Gamma$. The code $\calC$ is locally recoverable with locality $r$.
\end{lemma}

\begin{proof}
Given a function $f$ and codeword $c = (f(P))_{P \in \Gamma}$ and suppose we need to recover $f(P_0)$. We have that
$P_0 \in \Gamma_t$ for some $t$. Now, the restriction of $f$ to $E_t$ is a rational function $f_t$ on $E_t$,
which is an element of the Riemann-Roch space $\calL(E_t,r\calO)$. We claim that $f_t(P_0)$ can be uniquely recovered
from the values of $f_t(P), P \in \Gamma_t - \{P_0\}$. If there are two such functions with the same values, their
difference vanishes at $\Gamma_t - \{P_0\}$ but has a pole of order at most $r$ at $\calO$. The only possibilty is that this function then has simple
zeros at the points of $\Gamma_t - \{P_0\}$, a pole of order $r$ at $\calO$
and no other zeros or poles. That would imply, using Abel's theorem on $E_t$
(\cite[Corollary III 3.5]{Sil}), that
$\sum_{P \in \Gamma_t - \{P_0\}} P = \calO$ and thus $P_0 \in E_t[2]$, which contradicts our hypothesis. This shows that
the map $L(r\calO) \to \F_q^r, h \mapsto (h(P))_{P \in \Gamma_t - \{P_0\}}$ is injective. As these spaces have the same
dimension by Riemann-Roch, it is also surjective.
\end{proof}

A natural example is to take sections $P_i,i=1,\ldots,r$ of the elliptic fibration $\pi: \calE \to B$.
If we let $P_{r+1} = -\sum_{i=1}^r P_i$ and $\Gamma_t = \{P_1(t),\ldots, P_{r+1}(t)\}$,
we are in the above situation. 

We can also use an irreducible curve $C$ in $\calE$.
Then we have a map $C \to B$ and we assume that it has
degree $r+1$ and take as $\Gamma_t$ the fibers of this map above points that split completely. 
To ensure that the points of $\Gamma_t$ add to zero we need to check the algebraic point defined by $C$ has trace 
zero. Often the following lemma is useful.

\begin{lemma}
\label{trace}
Let $\pi: \calE \to B$ be an elliptic surface with finite Mordell-Weil group
of order prime to the characteristic of $k$. 
Let $C$ be an irreducible curve in $\calE$
such that the map $C \to B$ is separable of degree $r+1$.
If, for one $t \in B$ with $\pi^{-1}(t)$ an elliptic curve and whose preimage $\Gamma_t = (\pi{|_C})^{-1}(t)$ in $C$ has $r+1$ distinct points we have that 
$\sum_{P \in \Gamma_t} P = \calO$, then for all other such $t$, we also have $\sum_{P \in \Gamma_t} P = \calO$.
\end{lemma}

\begin{proof}
We can base change $\pi: \calE \to B$ to $\pi': \calE' \to C$ via $C \to B$ and $C$ itself pulls back to a section $s$ of $\pi'$
and we can then take the $C \to B$ trace of this section to get a section of $\pi$. Concretely, this section consists of adding
the points on $(\pi{|_C})^{-1}(t)$ and viewing that as a function of $t \in B$. By the assumption on the Mordell-Weil group, this section is of finite order. From \cite[Proposition VII 3.1]{Sil}, for sections of finite order prime to the characteristic, the specialization map to a smooth fiber
is injective. By assumption, for one such fiber, the specialization of $s$ is zero. It follows that $s$ itself is zero.
\end{proof}

Here are some explicit examples. 

\begin{example}
Take $\calE$ the Legendre family $y^2=x(x-1)(x-t)$ and consider the curve
$C: (u^2+t+1)^2=u(u-1)(u-t)$ of genus $1$ embedded in $\calE$ by taking $x=u,y=u^2+t+1$, so $r=3$. 
Lemma \ref{trace} applies with $t=-1$. If $\Gamma$ has $n$ points and
$d < n, 4|(n-d)$, we consider functions of the form $f=a(t)+b(t)x+c(t)y$ with $\deg a \le (n-d)/4, \deg b, \deg c < (n-d)/4$
and these restrict to $C$ as a function of degree at most $n-d$, so the minimum distance is at least $d$. The 
dimension is $k=3(n-d)/4 +1$ and it follows that 
$d = n -k  -\lceil k/3 \rceil +2$, i.e., the code is optimal, but 
typically not as long as the optimal codes from the previous sections. 
\end{example}

\begin{example}
Let $\calE$ be the elliptic surface $y^2=x^3+x-t^2-1$ over $\F_q$
and $C$ the curve
given by $x=y^2$ inside $\calE$, which is $y^6=t^2+1$. The elliptic surface
has trivial Mordell Weil group over $\F_q(t)$ so the multisection 
corresponding to $C$ automatically has trace zero. This leads to the same
family of codes corresponding to the case $r=5$ of subsection \ref{cyclic} by 
considering evaluation on functions of the form $f=a_0(t)+a_1(t)x+a_2(t)y+a_3(t)x^2+a_4(t)xy$.
\end{example}

\begin{example}
We can also recover the case $r=3$ of subsection \ref{cyclic} by taking
$\calE$ to be the elliptic surface $y^2+xy=x^3+t^2+2$ over $\F_q$
and $C$ the curve given by $x^2=y=u$ inside $\calE$, which is $u^4=t^2+2$ and
evaluation on functions of the form $f=a_0(t)+a_1(t)x+a_2(t)y$.
We can take, for $q=5,13$ respectively, sets of size 
$b=2,4$ and get codes of length $n=8,16$.
\end{example}

Yet another example is a variant of the examples constructed by Ulmer \cite{Ulmer} leading to the following theorem. 

\begin{theorem}
\label{ell-thm}
For every odd prime (power) $p$ and integer $d \le 2(p+1)(p-2), (p+1)|d$, 
there exists a locally recoverable code $\calC$ over $\F_{p^2}$ of recoverability $p$, length $n=2(p+1)(p-2)$,
minimum distance $d$ and dimension 
$$k = \frac{p(n-d)}{p+1}-\frac{p-1}{2}.$$
\end{theorem}

\begin{proof}
Consider the surface $\calE: y^2=x(x+1)(x+t^2+1)$ over $\F_{p^2}$, $p$ odd and the curve $C$ 
defined by $u^{p+1} = t^2+1$. Then $C$ embeds in $\calE$ by taking $x = u, y = u(u+1)^{(p+1)/2}$.
The points on $C$ on the fiber above $t=b$ are of the form
$(c, c(c+1)^{(p+1)/2})$ for each $c$ satisfying $c^{p+1} = b^2+1$. The function $y(x+1)^{(p-1)/2}-(x+b^2+1)$ has degree
$p+2$ and vanishes on all these points and on the point $(-b^2-1,0)$ of order $2$. So lemma \ref{ell-recover}
applies once we exclude the points on $C$ with $c=0, c^{p+1}=1$. Each allowed value of $c$ gives two values of
$b$ since $c^{p+1}-1 \in \F_p$ so has square roots in $\F_{p^2}$. So we have $n=2(p+1)(p-2)$ points in $C$ we can
use to form $\Gamma$.

To construct a code we consider the following vector space, where $x_i = x^{(i+1)/2}, i$ odd and $x_i = yx^{(i-2)/2}, i$ even, $i>0$.

\[
V = \left\{ a_0(t) + \sum_{i = 1}^{p-1} a_i(t)x_i : \deg a_i \leq N_i \text{ for } i = 0,\dots,p-1\right\}
\]

where $N_0 = \frac{n-d}{p+1}$,

\begin{equation}
\begin{split}
N_i = \frac{n-d}{p+1} -1,\, i\, \text{ odd },\\
N_i = \frac{n-d }{p+1} -2,\, i\, \text{ even }, i > 0.
\end{split}
\end{equation}

\noindent
chosen so that the elements of $V$ restrict to functions of degree $n-d$ on $C$ and the codewords have weight
at least $d$. The dimension $k$ satisfies $k = \sum_{i=0}^{p-1} (N_i+1)$ and the result follows.
\end{proof}

\begin{remark}
Note that, in the above theorem $d_{\opt} = n-k-\lceil k/p \rceil +2 = d + (p+3)/2$.
\end{remark}

\section*{Acknowledgements}
The authors would like to thank the following institutions for providing the opportunity for them to meet and/or for financial
support: IMPA, BIRS-Oaxaca, MPIM, IHP, University of Canterbury and M. Stoll's Rational Points workshop series. Cec\'ilia Salgado was partially supported by FAPERJ grant E-26/203.205/2016,  the Serrapilheira Institute (grant number Serra-1709-17759), Cnpq grant PQ2 310070/2017-1 and the Capes-Humboldt program. Anthony V\'arilly-Alvarado was partially supported by NSF grants DMS-1352291 and DMS-1902274. 
Jos\'e Felipe Voloch was partially supported by the Simons Foundation (grant \#234591) and the Marsden Fund Council administered by the Royal Society of New Zealand. He would also like to thank A. Dimakis for a helpful conversation.
{Finally, the authors would like to thank the referees for a careful reading
and a number of suggestions to the manuscript.}



\begin{bibdiv}
\begin{biblist}

\bib{AtiyahMacdonald}{book}{
   author={Atiyah, M. F.},
   author={Macdonald, I. G.},
   title={Introduction to commutative algebra},
   publisher={Addison-Wesley Publishing Co., Reading, Mass.-London-Don
   Mills, Ont.},
   date={1969},
   pages={ix+128},
}
		
\bib{BargEtAl2017}{article}{
   author={Barg, A.},
   author={Haymaker, K.},
   author={Howe, E. W.},
   author={Matthews, G. L.},
   author={V\'{a}rilly-Alvarado, A.},
   title={Locally recoverable codes from algebraic curves and surfaces},
   conference={
      title={Algebraic geometry for coding theory and cryptography},
   },
   book={
      series={Assoc. Women Math. Ser.},
      volume={9},
      publisher={Springer, Cham},
   },
   date={2017},
   pages={95--127},
}

\bib{BargTamoEtAl2015}{article}{
   author={Barg, A.},
   author={Tamo, I.},
   author={Vl\u{a}du\c{t}, S.},
   title={Locally recoverable codes on algebraic curves},
   journal={IEEE Trans. Inform. Theory},
   volume={63},
   date={2017},
   number={8},
   pages={4928--4939},
   issn={0018-9448},
}

\bib{Coz}{article}{
	author = {Coskun, I.},
	AUTHOR = {Huizenga, S.},
	Title= {Brill-Noether theorems and globally generated vector bundles on Hirzebruch surfaces},
	year = {2017},
	note = {Preprint, arXiv:1705.08460},
	}

\bib{GHM}{article}{
   author={Galindo, C.},
   author={Hernando, F.},
   author={Munuera, C.},
   title={Locally recoverable $J$-affine variety codes},
   journal={Finite Fields Appl.},
   volume={64},
   date={2020},
   pages={101661, 22},
   issn={1071-5797},
}

\bib{GopalanHuangEtAl2012}{article}{
   author={Gopalan, P.},
   author={Huang, C.},
   author={Simitci, H.},
   author={Yekhanin, S.},
   title={On the locality of codeword symbols},
   journal={IEEE Trans. Inform. Theory},
   volume={58},
   date={2012},
   number={11},
   pages={6925--6934},
   issn={0018-9448},
}

\bib{Guruswami2018}{article}{
   author={Guruswami, V.},
   author={Xing, C.},
   author={Yuan, C.},
   title={How long can optimal locally repairable codes be?},
   journal={IEEE Trans. Inform. Theory},
   volume={65},
   date={2019},
   number={6},
   pages={3662--3670},
   issn={0018-9448},
}

\bib{Hartshorne}{book}{
   author={Hartshorne, R.},
   title={Algebraic geometry},
   note={Graduate Texts in Mathematics, No. 52},
   publisher={Springer-Verlag, New York-Heidelberg},
   date={1977},
   pages={xvi+496},
   isbn={0-387-90244-9},
}

\bib{Malmskog2016}{article}{
   author={Haymaker, K.},
   author={Malmskog, B.},
   author={Matthews, G. L.},
   title={Locally recoverable codes with availability $t\geq 2$ from fiber
   products of curves},
   journal={Adv. Math. Commun.},
   volume={12},
   date={2018},
   number={2},
   pages={317--336},
   issn={1930-5346},
}

\bib{MT}{article}{
   author={Munuera, C.},
   author={Ten\'{o}rio, W.},
   title={Locally recoverable codes from rational maps},
   journal={Finite Fields Appl.},
   volume={54},
   date={2018},
   pages={80--100},
   issn={1071-5797},
}

\bib{MTT}{article}{
   author={Munuera, C.},
   author={Ten\'{o}rio, W.},
   author={Torres, F.},
   title={Locally recoverable codes from algebraic curves with separated
   variables},
   journal={Adv. Math. Commun.},
   volume={14},
   date={2020},
   number={2},
   pages={265--278},
   issn={1930-5346},
}

\bib{Xing}{article}{
   author={Li, X.},
   author={Ma, L.},
   author={Xing, C.},
   title={Optimal locally repairable codes via elliptic curves},
   journal={IEEE Trans. Inform. Theory},
   volume={65},
   date={2019},
   number={1},
   pages={108--117},
   issn={0018-9448},
}

\bib{Dimakis}{article}{
   author={Papailiopoulos, D. S.},
   author={Dimakis, A. G.},
   title={Locally repairable codes},
   journal={IEEE Trans. Inform. Theory},
   volume={60},
   date={2014},
   number={10},
   pages={5843--5855},
   issn={0018-9448},
}

\bib{Reid}{article}{
   author={Reid, M.},
   title={Chapters on algebraic surfaces},
   conference={
      title={Complex algebraic geometry},
      address={Park City, UT},
      date={1993},
   },
   book={
      series={IAS/Park City Math. Ser.},
      volume={3},
      publisher={Amer. Math. Soc., Providence, RI},
   },
   date={1997},
   pages={3--159},
}

\bib{Sil}{book}{
author={Silverman, J.}
title={The arithmetic of elliptic curves},
   series={Graduate Texts in Mathematics},
   volume={106},
   edition={2},
   publisher={Springer, Dordrecht},
   date={2009},
}



\bib{TamoBarg2014}{article}{
   author={Tamo, I.},
   author={Barg, A,},
   title={A family of optimal locally recoverable codes},
   journal={IEEE Trans. Inform. Theory},
   volume={60},
   date={2014},
   number={8},
   pages={4661--4676},
   issn={0018-9448},
}

\bib{AGBook}{book}{
   author={Tsfasman, M.},
   author={Vl\u{a}du\c{t}, S.},
   author={Nogin, D.},
   title={Algebraic geometric codes: basic notions},
   series={Mathematical Surveys and Monographs},
   volume={139},
   publisher={American Mathematical Society, Providence, RI},
   date={2007},
   pages={xx+338},
   isbn={978-0-8218-4306-2},
}

\bib{Ulmer}{article}{
   author={Ulmer, D.},
   title={Explicit points on the Legendre curve},
   journal={J. Number Theory},
   volume={136},
   date={2014},
   pages={165--194},
   issn={0022-314X},
}
\end{biblist}
\end{bibdiv}
\end{document}